\pdfoutput=1
\documentclass{article}
\usepackage{spconf}

\usepackage[T1]{fontenc}
\usepackage[utf8]{inputenc}
\usepackage{amsmath,amssymb,amsthm}
\usepackage{mathtools}
\usepackage{bm}
\usepackage{booktabs,array}
\usepackage{graphicx}
\usepackage{balance}
\usepackage[inline]{enumitem}

\usepackage{siunitx}
\usepackage{caption}
\usepackage{subcaption}
\usepackage[%
  backend = biber,%
  maxbibnames = 10,%
  minbibnames = 2,%
  style = ieee,%
  citestyle = numeric-comp,%
  bibencoding = utf8,%
  sortcites = true,%
  isbn = false,%
  url = false,%
  doi = false,
  eprint = true,%
]{biblatex}

\usepackage{software-biblatex}

\usepackage[ruled,linesnumbered]{algorithm2e}

\makeatletter
\def\@ssect#1#2#3#4#5{\begingroup\bf\centering
  {\interlinepenalty\@M\uppercase{#5}\par}\endgroup\@xsect{#3}}
\renewcommand\section{\@startsection{section}{1}{\z@}%
  {-1\baselineskip \@plus -.5ex \@minus -.2ex}{1\baselineskip}{\normalfont\bfseries}}
\renewcommand\subsection{\@startsection{subsection}{2}{\z@}%
  {-1\baselineskip \@plus -.5ex \@minus -.2ex}{.3ex}{\normalfont\bfseries}}

\defbibheading{bibliography}{%
  \addvspace{1\baselineskip}
  \@ssect{\z@}{-1\baselineskip \@plus -.5ex \@minus
    -.2ex}{1\baselineskip}{\normalfont\bfseries}{REFERENCES}%
}

\makeatother
\usepackage[hidelinks]{hyperref}

\usepackage{cleveref}
\crefname{thm}{theorem}{theorems}
\Crefname{thm}{Theorem}{Theorems}
\crefname{prop}{proposition}{propositions}
\Crefname{prop}{Proposition}{Propositions}

\crefname{section}{section}{sections}
\Crefname{section}{Section}{Sections}

\graphicspath{{figures/}}

\newtheorem{proposition}{Proposition}

\usepackage{cleveref}
\crefname{section}{Sec.}{Secs.}
\Crefname{section}{Section}{Sections}
\crefname{subsection}{Sec.}{Secs.}
\Crefname{subsection}{Section}{Sections}
\crefname{figure}{Fig.}{Figs.}
\Crefname{figure}{Fig.}{Figs.}
\crefname{table}{Table}{Tables}
\Crefname{table}{Table}{Tables}
\crefname{proposition}{Prop.}{Props.}
\Crefname{proposition}{Proposition}{Propositions}
\AtBeginDocument{%
  \crefname{algorithm}{Algorithm}{Algorithms}%
  \Crefname{algorithm}{Algorithm}{Algorithms}}
\crefformat{equation}{(#2#1#3)}
\Crefformat{equation}{(#2#1#3)}
\crefrangeformat{equation}{(#3#1#4)--(#5#2#6)}
\Crefrangeformat{equation}{(#3#1#4)--(#5#2#6)}
\crefmultiformat{equation}{(#2#1#3)}{ and (#2#1#3)}{, (#2#1#3)}{, and (#2#1#3)}
\Crefmultiformat{equation}{(#2#1#3)}{ and (#2#1#3)}{, (#2#1#3)}{, and (#2#1#3)}

\usepackage{tikz}
\definecolor{icfixedfill}{HTML}{DCE6F2}
\definecolor{icfixedline}{HTML}{2D4C73}
\definecolor{ictrainfill}{HTML}{F6DEDE}
\definecolor{ictrainline}{HTML}{9C3030}
\definecolor{icwire}{HTML}{7A8BA3}
\newcommand{\icx}{\tikz[baseline=-0.55ex]{%
  \draw[icwire,line width=0.5pt]
    (0.10,0.085) -- (0.28,0.085) (0.10,0) -- (0.28,0)
    (0.10,-0.085) -- (0.28,-0.085);
  \node[draw,line width=0.6pt,rounded corners=1.2pt,fill=white,
    inner sep=1.3pt,font=\tiny] at (0,0) {$\mathbf{x}$};}}
\newcommand{\icF}{\tikz[baseline=-0.55ex]{%
  \draw[icwire,line width=0.5pt]
    (-0.28,0.085) -- (0.28,0.085) (-0.28,0) -- (0.28,0)
    (-0.28,-0.085) -- (0.28,-0.085);
  \node[draw=icfixedline,line width=0.6pt,rounded corners=1.2pt,fill=white,
    inner sep=1.1pt,font=\tiny] at (0,0) {$\mathbf{F}_{N}$};}}
\newcommand{\icH}{\tikz[baseline=-0.55ex]{%
  \draw[icwire,line width=0.5pt] (-0.22,0) -- (0.22,0);
  \node[draw=icfixedline,line width=0.6pt,rounded corners=1.2pt,
    fill=icfixedfill,inner sep=1.3pt,font=\tiny] at (0,0) {$\mathbf{H}$};}}
\newcommand{\icCP}{\tikz[baseline=-0.55ex]{%
  \draw[icwire,line width=0.5pt] (-0.28,0) -- (0.28,0);
  \draw[icfixedline,line width=0.6pt] (0,-0.16) -- (0,0);
  \fill[icfixedline] (0,-0.16) circle (0.8pt);
  \node[draw=icfixedline,line width=0.6pt,rounded corners=1.2pt,
    fill=icfixedfill,inner sep=1.1pt,font=\tiny] at (0,0) {$\bm{\Phi}$};}}
\newcommand{\icM}{\tikz[baseline=-0.55ex]{%
  \draw[icwire,line width=0.5pt] (-0.24,0) -- (0.24,0);
  \draw[ictrainline,line width=0.6pt] (0,-0.16) -- (0,0);
  \fill[ictrainline] (0,-0.16) circle (0.8pt);
  \node[draw=ictrainline,line width=0.6pt,rounded corners=1.2pt,
    fill=ictrainfill,inner sep=1.3pt,font=\tiny] at (0,0) {$\bm{\Phi}$};}}

\definecolor{orcidgreen}{HTML}{A6CE39}
\newcommand{\orcidlink}[1]{\href{https://orcid.org/#1}{\tikz[baseline=-3.2pt]{%
  \fill[orcidgreen] (0,0) circle (6.5pt);
  \node[white,inner sep=0,font=\fontsize{9}{9}\selectfont\sffamily\bfseries]
    at (0.3pt,0) {iD};}}}

\makeatletter
\def\tightdisplayskips{%
  \setlength{\abovedisplayskip}{2pt plus 2pt minus 1pt}%
  \setlength{\belowdisplayskip}{2pt plus 2pt minus 1pt}%
  \setlength{\abovedisplayshortskip}{1pt plus 2pt}%
  \setlength{\belowdisplayshortskip}{1pt plus 2pt minus 1pt}}
\g@addto@macro\normalsize{\tightdisplayskips}
\g@addto@macro\small{\tightdisplayskips}
\makeatother

\newcommand{\R}{\mathbb{R}}

\newcommand{\Herm}{\mathsf{H}}
\newcommand{\Tr}{\mathsf{T}}
\newcommand{\Net}{\mathcal{T}}
\newcommand{\A}{\mathcal{A}}
\newcommand{\Hard}{\mathcal{H}}
\newcommand{\norm}[1]{\lVert #1\rVert}

\title{QUANTUM-INSPIRED TRAINABLE AND PARAMETER-EFFICIENT\\ TENSOR NETWORKS FOR IMAGE
  INPAINTING\vspace{-10pt}}

\name{Shiwen An\,\orcidlink{0000-0002-4727-2906} and Konstantinos
  Slavakis\,\orcidlink{0000-0002-3370-3154}\vspace{-10pt}}

\address{\small Institute of Science Tokyo, Department of Information and Communications
  Engineering, Yokohama, Japan\vspace{-10pt}}

\begin{document}
\ninept
\maketitle


\begin{abstract}
  This work introduces quantum-inspired tensor-network circuits as trainable transforms for image inpainting. Among the proposed architectures, the diagonal quantum Fourier transform (QFT) relaxation is invertible with $O(N^2 \log N)$ computational cost for $N\times N$ images, inherently preserving minimum coherence throughout training via its circuit structure and eliminating the need for explicit coherence penalties. Unconstrained gradient-based phase optimization (Riemannian-optimization free) enables efficient learning from randomly sampled training data, allowing the learned transform to generalize to test images observed through fixed sampling masks. Numerical tests show that the learned models outperform fixed transforms and per-image optimization while matching the performance of much larger unitary architectures, yet with far fewer parameters. \end{abstract}


\section{Introduction}
\label{sec:intro}

Image inpainting reconstructs a digital image from a partial set of observed entries, inferring missing
values by assuming a prior model of the image's structure. \textit{Sparse recovery}\/ exploits
transform-domain sparsity models, using wavelets~\cite{mallat2008} or learned
dictionaries~\cite{aharon2006, ravishankar2013closed}, with recovery guarantees~\cite{candes2006,
  blumensath2009, guleryuz2006}. Low-rank matrix completion capitalizes on row-column correlations in the
image matrix, with theoretical guarantees under incoherence and sampling
conditions~\cite{candes2009matrix, candes2010power}. Deep-learning approaches include trained inpainting
networks~\cite{suvorov2022lama}, diffusion-based methods~\cite{lugmayr2022repaint}, and plug-and-play
solvers with learned denoising priors~\cite{sreehari2016pnp}. Other methods optimize representations
\textit{directly on individual test images,} including deep image priors~\cite{ulyanov2018} and
tensor-train approaches, such as low-rank decomposition~\cite{bengua2017tt} and coarse-to-fine
refinement~\cite{loeschcke2024}.

For \textit{transform-based recovery,} sparsity alone does not determine which transform to use. A basis
atom (column vector) of the transform localized to a few pixels does not appear in the transform-domain
signal representation whenever samples miss its support. Coherence $\mu(\cdot)$ measures the squared
maximum modulus of correlation between basis atoms and the standard basis---see
\eqref{eq:coh}---quantifying the localization of atoms in the image domain. The complexity of recovery is
known to grow with coherence~\cite{candes2007}---a fundamental tension in compressed
sensing~\cite{adcock2017}. Empirically, wavelets (highly localized bases, thus high coherence) achieve
sparser approximations than the DFT, yet produce worse inpainting results under the same solver---see
\cref{sec:exp}. This gap motivates the design of \textit{trainable transforms with low coherence.}  Since
sparsity also varies across bases, coherence is an indicative but not the sole factor determining
inpainting performance.

Tensor networks offer a quantum-inspired and compact way to parameterize transforms---see
\cref{fig:net}. Writing pixel indices in binary represents an image as a tensor with one mode per index
bit~\cite{oseledets2011, khoromskij2011}. A tensor network of unitary gates on these modes defines a
basis shared across images. The basis retains two key properties of the FFT: fast application with
$O(N^2\log N)$ for $N\times N$ images, and invertibility via its adjoint. The quantum Fourier transform
(QFT) circuit~\cite{nielsen2010} achieves both through a factorization into Hadamard and controlled-phase
gates. This circuit relaxes the diagonal gates while retaining one Hadamard per bit. Every resulting
matrix yields coherence $\mu = 1$, ensuring minimum coherence and unitarity under adaptation. Learnable
butterfly factorizations~\cite{dao2019} also adapt a fast transform, but their relaxed blocks do not
guarantee minimum coherence.

Building on QFT-gate relaxations developed for compression~\cite{pdft2026}, where tensor networks were
trained via Riemannian optimization, this work extends those circuits into image inpainting. More
specifically, the proposed contributions are threefold.

\begin{enumerate}[label = \textbf{(\alph*)}, leftmargin = *, wide, topsep = 0pt, itemsep = 0pt, parsep =
    0pt]

\item \textbf{First application to image inpainting.} Quantum-inspired tensor-network circuits are
  adapted to image inpainting via gradient-based phase optimization for both phase-only and diagonal
  relaxations (Riemannian-optimization free unlike~\cite{pdft2026})---see \cref{sec:method}.

\item \textbf{Guaranteed coherence.} The diagonal QFT relaxation and its phase-only subfamily guarantee
  minimum coherence throughout training via the circuit topology, without requiring an explicit coherence
  penalty. This structural guarantee distinguishes the proposed tensor networks from both unconstrained
  learned dictionaries and from QFT circuits with free Hadamards, whose coherence can exceed the minimum
  bound (\cref{prop:flat}).

\item \textbf{Parameter-efficiency and performance.} The proposed learned diagonal model outperforms
  fixed transform baselines and per-image optimization schemes while matching the performance of much
  larger unitary architectures, yet with far fewer parameters (\cref{sec:exp}).

\end{enumerate}

\begin{figure*}[t]
  \centering
  \input{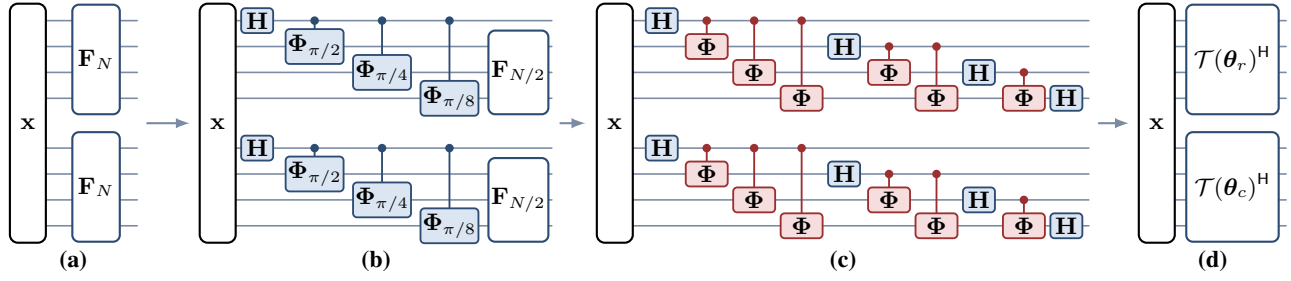}
  \caption{The construction of \cref{sec:family} for $n=4$ per axis, depicted as four tensor networks
    progressing from specific to general (left to right), each applied to a copy of the vectorized
    image $\mathbf{x} \coloneqq \mathrm{vec}(\mathbf{X})$. %
    \textbf{(a)}~The separable 2D DFT, $(\mathbf{F}_N \otimes \mathbf{F}_N)\, \mathbf{x}$, from
    Step~1a. %
    \textbf{(b)}~One Cooley--Tukey level: the radix-2 decomposition \cref{eq:radix2} of Step~1b. %
    \textbf{(c)}~The circuit $\Net(\bm{\theta})$ of \cref{eq:pdft} up to the bit reversal $\mathbf{B}$, with trainable
    diagonals (red) and fixed Hadamards (blue). It is unitary by \cref{eq:unitary} and satisfies
    $\mu = 1$ at every $\bm{\theta}$ by \cref{prop:flat}. %
    \textbf{(d)}~The analysis map $\A_{\bm{\theta}}$ of Step~1d.}%
  \label{fig:net}
\end{figure*}

\section{The Image-Inpainting Problem}
\label{sec:problem}

An image $\mathbf{X} \in \R^{N\times N}$ is observed partially, with entries known only on a subset of
pixel indices $\Omega$. The observation operator $P_{\Omega} \colon \R^{N\times N} \to \R^{N\times N}$
models this by $\mathbf{Y} \coloneqq P_{\Omega} \mathbf{X} \coloneqq \mathbf{M}_{\Omega} \odot
\mathbf{X}$, where $\mathbf{M}_{\Omega} \in \{0, 1\}^{N\times N}$ is a binary mask with ones at observed
locations $\Omega$ and zeroes elsewhere, and $\odot$ denotes the Hadamard product. The sampling mask
$\mathbf{M}_{\Omega}$ is viewed as a random variable (RV) to account for all possible sampling
patterns. Typically, the entries of $\mathbf{M}_{\Omega}$ are modeled as independent Bernoulli RVs with
$p$ being the probability of success (appearance of $1$s), also called \textit{sampling rate,} so that
$pN^2$ pixels are observed on average. Image inpainting recovers $\mathbf{X}$ from the partial
observations $\mathbf{Y}$.

This work introduces tensor networks $\Net( \cdot )$---see \cref{fig:net}---for transform-based image
inpainting. Unlike fixed classical transforms such as the DCT, the parameters $\bm{\theta}$ of $\Net(
\bm{\theta} )$ are learned from training data $\mathcal{D}$ by minimizing the loss $\mathcal{L}(\cdot)$:
\begin{align}
  \bm{\theta}_{\star} & \in \operatorname*{arg\,min} \nolimits_{ \bm{\theta}}\ \mathcal{L} ( \bm{\theta}
  ), \notag \\
  \mathcal{L} ( \bm{\theta} ) & \coloneqq \frac{1}{ \lvert \mathcal D \rvert } \sum\nolimits_{ \mathbf{X}
    \in\mathcal D} \mathbb{E}_{\Omega} \! \left \{ \tfrac{1}{N^2} \norm{ \widehat{\mathbf{X}}_K (
    \bm{\theta}; \mathbf{Y}, \Omega) - \mathbf{X} }_{\text{F}}^2 \right \} \,, \label{eq:task}
\end{align}
where expectation $\mathbb{E}_{\Omega}\{ \cdot \}$ averages over random sampling masks
$\mathbf{M}_{\Omega}$, and $\widehat{ \mathbf{X} }_K( \bm{\theta}; \mathbf{Y}, \Omega)$ denotes the
$K$-step recovered image from \cref{alg:train}. To manage computational budgets, stochastic/online
gradient descent solves \eqref{eq:task} by sampling a small batch of training images per step (to
approximate the sum over all training data) and randomly drawing a sampling mask per image (to
approximate $\mathbb{E}_{\Omega}\{ \cdot \}$)---see \cref{alg:train}. Once $\bm{\theta}_{\star}$ is
learned, the network $\Net( \bm{\theta}_{\star} )$ inpaints a test image $\mathbf{X}_{\text{test}}
\notin \mathcal{D}$ observed through a fixed sampling mask $\Omega_{\text{test}}$ via the
$\texttt{Recover} ( P_{ \Omega_{\text{test}} } \mathbf{X}_{\text{test}}, \Omega_{\text{test}},
\bm{\theta}_{\star} )$ function of \cref{alg:train}.

As explained in \cref{sec:intro}, coherence $\mu(\cdot)$ quantifies the localization of basis atoms. For
$\mathbf{U} = [ u_{ij} ] \in U(N)$ (where $U(N)$ is the set of complex-valued $N\times N$ unitary
matrices), coherence is defined as:
\begin{equation}
  \mu( \mathbf{U} ) \coloneqq N \max_{i,j}\, \lvert \mathbf{e}_i^{\Tr} \mathbf{u}_j  \rvert^2 = N
  \max_{i,j}\, \lvert u_{ij}  \rvert^2 \in[1,N] \,, \label{eq:coh}
\end{equation}
where $\{ \mathbf{e}_i \}_{i=1}^N$ is the standard basis of $\R^{N}$ and $\mathbf{u}_j$ is the $j$th
column of $\mathbf{U}$. Under sparsity and suitable sampling conditions, the sample complexity for
recovery grows linearly with $\mu$~\cite{candes2007}, motivating the design of transforms with low
coherence. For images---which have both row and column structure---the 2D-coherence $\mu_{\mathrm{2D}}$
is defined as the product of coherences along each dimension. Correspondingly, the tensor-network
parameter vector $\bm{\theta} = ( \bm{\theta}_r, \bm{\theta}_c)$ factorizes into row and column
components, respecting this 2D structure.

\section{The Tensor Networks}
\label{sec:method}

This section writes the DFT as a quantum Fourier circuit (\cref{sec:family}), frees its gates one
structure at a time (\cref{sec:shared}), trains the result through the recovery iteration
(\cref{sec:unroll}), and shows the two properties the design rests on, isometry at every parameter value
(\cref{sec:iso}) and minimum coherence (\cref{sec:pinned}), which single out QFT (diagonals) as the
recommended model.

\subsection{An FFT-factorized isometric family}
\label{sec:family}

Read the index of a length-($N=2^n$) axis ($n \ge 2$) as an $n$-bit string $b_0 b_1 \cdots b_{n-1}$,
$b_0$ the most significant bit, and draw one wire per bit; the $n$ wires of an axis form a register. The
network acts on the vectorized image $\mathbf{x} = \mathrm{vec}( \mathbf{X} )$ \icx{}, the rows of
$\mathbf{X}$ stacked, and a transform is a network of small unitary
tensors on those wires~\cite{nielsen2010, pdft2026}. Throughout, $\otimes$ is the Kronecker product,
and every product of gates is the ordinary matrix product, written in the order applied, the rightmost
first. \cref{fig:net} draws the four steps.

\textbf{Step 1a.} The separable 2D DFT attaches one dense $\mathbf{F}_N$ \icF{}, $( \mathbf{F}_N )_{jk}
= e^{2\pi ijk/N} / \sqrt{N}$, to each register: $\mathbf{X} \mapsto \mathbf{F}_N \mathbf{X}
\mathbf{F}_N^{\Tr}$, i.e., $( \mathbf{F}_N \otimes \mathbf{F}_N)\, \mathbf{x}$. No gate joins the two.

\textbf{Step 1b.} The Hadamard on wire $q$ is $\mathbf{H}_q = \mathbf{I}_{2^{q}} \otimes \mathbf{H}
\otimes \mathbf{I}_{2^{\,n-1-q}}$ \icH{}, with $\mathbf{H} \coloneqq \frac{1}{\sqrt{2}} \bigl[
\begin{smallmatrix} 1 & 1 \\ 1 & -1 \end{smallmatrix} \bigr]$. The controlled phase on the wire pair
$(p, q)$, $\bm{\Phi}_{pq} (\xi)$ \icCP{}, is the phase gate $\mathrm{diag} (1, 1, 1, e^{i\xi})$ on wires
$p$ and $q$ and $\mathbf{I}_2$ on every other wire. One level of the
Cooley--Tukey recursion factors each dense block into a Hadamard on the most significant bit, a
controlled phase coupling it to each remaining bit $p$ at the angle $\pi/2, \pi/4, \ldots, 2\pi/2^{n}$
in turn, and a half-size $\mathbf{F}_{N/2}$ on the rest, the angle of the pair $(p, q)$ being $\xi_{pq} = 2\pi/2^{\,p-q+1}$. In matrix form this level is
the classical radix-2 identity~\cite{vanloan1992}
\begin{equation}
  \mathbf{F}_N = \mathbf{S}\,( \mathbf{I}_2 \otimes \mathbf{F}_{N/2}) \prod_{p=1}^{n-1} \bm{\Phi}_{p0}
  (\xi_{p0})\, \mathbf{H}_0 \,, \label{eq:radix2}
\end{equation}
with $\mathbf{S}$ the permutation matrix that moves bit $0$ of the index to the last place, $b_0 b_1
\cdots b_{n-1} \mapsto b_1 \cdots b_{n-1} b_0$; the $\bm{\Phi}$ factors are diagonal and commute.

\textbf{Step 1c.} Recursing until $\mathbf{F}_2 = \mathbf{H}$ leaves nothing dense. Level $q$ contributes
$\mathbf{H}_q$ followed by its controlled phases $\mathbf{D}_q = \prod_{p=q+1}^{n-1} \bm{\Phi}_{pq}
(\xi_{pq})$, and its move of bit $q$ to the last place, $\mathbf{S}_q = \mathbf{I}_{2^q} \otimes
\mathbf{S}$ with $\mathbf{S}$ on the trailing $n-q$ bits; the moves compose into the bit reversal
$\mathbf{B} = \mathbf{S}_0 \mathbf{S}_1 \cdots \mathbf{S}_{n-2}$. The levels compose in circuit order,
the $q=0$ factors rightmost and applied first: $\mathbf{F}_N = \mathbf{B}\, ( \mathbf{D}_{n-1}
\mathbf{H}_{n-1} ) \cdots ( \mathbf{D}_0 \mathbf{H}_0)$, the quantum Fourier circuit~\cite{nielsen2010}.

\textbf{Step 1d.} Freeing the angle of every controlled phase, $\xi_{pq} \mapsto \theta_{pq}$, gives the
trainable network
\begin{equation}
  \Net( \bm{\theta} ) \coloneqq \mathbf{B}\, \big( \mathbf{D}_{n-1} ( \bm{\theta} ) \mathbf{H}_{n-1}
  \big) \cdots \big( \mathbf{D}_0 ( \bm{\theta} ) \mathbf{H}_0 \big) \,, \label{eq:pdft}
\end{equation}
with $\mathbf{D}_q ( \bm{\theta} ) = \prod_{p=q+1}^{n-1} \bm{\Phi}_{pq} (\theta_{pq})$, $\bm{\theta}$ the
parameter vector of \cref{sec:problem} (its entries per model in \cref{sec:shared}), and $\Net(
\bm{\theta}^0 ) = \mathbf{F}_N$ at $\theta_{pq} = \xi_{pq}$. Every factor is unitary at every parameter
value, so
\begin{equation}
  \Net( \bm{\theta} )^{\Herm} \Net( \bm{\theta}) = \mathbf{I}\,, \quad\forall \bm{\theta} \,.
  \label{eq:unitary}
\end{equation}
The analysis map applies $\Net( \bm{\theta} )^\Herm$ on each register (\cref{fig:net}(d)); with
separate angles per axis,
\begin{equation}
  \A_{ \bm{\theta} } ( \mathbf{X} ) \coloneqq \Net( \bm{\theta}_r )^{\Herm} \mathbf{X}\, \overline{\Net(
    \bm{\theta}_c )}\,,\, \A_{ \bm{\theta} }^{-1} ( \mathbf{C} ) = \Net( \bm{\theta}_r ) \,
  \mathbf{C} \, \Net( \bm{\theta}_c )^{\Tr} \,, \label{eq:2d}
\end{equation}
the second being the exact inverse of the first by \cref{eq:unitary}.

\subsection{A ladder of relaxations}
\label{sec:shared}

Each model frees one structure of \cref{eq:pdft}, and $\bm{\theta}$ collects the free parameters of
whichever model is meant:
\par\smallskip
\setlength{\tabcolsep}{2.5pt}%
\centerline{%
  \resizebox{\columnwidth}{!}{%
    \begin{tabular}{@{}lccc@{}}
      \toprule
      & QFT (phases) & QFT (diagonals) & QFT (rotations)\\
      \midrule
      pair diagonal & $(1, 1, 1, e^{i\xi})$ & $(e^{i\theta_1}, \ldots, e^{i\theta_4})$ &
      $(e^{i\theta_1}, \ldots, e^{i\theta_4})$\\
      wire gate & $\mathbf{H}$ & $\mathbf{H}$ & $\mathbf{V} \in U(2)$\\
      count per axis & $n(n-1)/2$ & $2n(n-1)$ & $2n(n-1)+4n$\\
      \bottomrule
    \end{tabular}
}} \smallskip\noindent%
QFT (phases) is \cref{eq:pdft} as written. QFT (diagonals) frees the other three
phases of each pair, one per value $(b_p, b_q)$ of the two bits \icM{}, initialized at $(0, 0, 0,
\xi_{pq})$, and pinning them at zero recovers QFT (phases). QFT (rotations) frees the Hadamards too, each
$\mathbf{V}_q$ initialized at $\mathbf{H}$ and held on $U(2)$ by a Cayley retraction~\cite{wen2013} as
in~\cite{pdft2026}, and can leave the complex Hadamard set. Every replacement is unitary, so
\cref{eq:unitary} holds for all three, and all start at the DFT and minimize \cref{eq:task}, each
adjacent comparison in \cref{tab:reversal}'s last block isolating one freedom.

Tying the phase vectors of QFT (diagonals) by gate distance, $\bm{\theta}_{pq} = \bm{\psi}_{p-q}$ as in
the DFT rule $\xi_{pq}$ of Step~1b, gives $4(n-1)$ parameters per axis and, with the DFT phases at new
distances, a basis at every resolution within \cref{prop:flat}'s hypothesis (\cref{sec:pinned}).

\SetAlgoNlRelativeSize{0}
\begin{algorithm}[t]
  \SetCommentSty{textrm}
  \DontPrintSemicolon
  \SetAlgoNoEnd
  \SetKwProg{Fn}{function}{:}{}
  \SetKwFunction{FRec}{Recover}
  \KwIn{$\mathcal{D}$, $p$, $k$, $K$, $T$ (\cref{sec:exp})}
  \Fn{\FRec{$\mathbf{Y}, \Omega, \bm{\theta}$}}{
    $\mathbf{X}^{(0)} \gets \mathbf{Y}$\;
    \lFor{$\kappa = 0, \dots, K-1$}{\\
      Update $\mathbf{X}^{(\kappa + 1)}$ by \eqref{eq:iht}
    }
    \KwRet $\widehat{ \mathbf{X} }_K( \bm{\theta}; \mathbf{Y}, \Omega) \coloneqq \mathbf{X}^{(K)}$\;
  }
  $\bm{\theta} \gets \bm{\theta}^0$\tcp*[r]{the DFT}
  \For{$\tau =1, \dots, T$}{
    Draw $\mathbf{X} \in \mathcal{D}$ and a fresh mask as in \cref{sec:problem}\;
    $\widehat{ \mathbf{X} }_K \gets{}$\FRec{$P_{\Omega} \mathbf{X}, \Omega, \bm{\theta}$}\;
    $\bm{\theta} \gets \mathrm{Adam} (\bm{\theta}, \nabla_{ \bm{\theta} } \norm{\widehat{ \mathbf{X} }_K
      - \mathbf{X} }_{ \text{F} }^2 / N^2 )$
    \tcp*[r]{support fixed}
  }
  \KwOut{$\bm{\theta}_{\star} \gets \bm{\theta}$, with $\mu_{\mathrm{2D}} = 1$
    (\cref{prop:flat})}
  \caption{Recovery and training for the phase-only and diagonal models}
  \label{alg:train}
\end{algorithm}

\subsection{Recovery as an unrolled map}
\label{sec:unroll}

Define the hard-thresholder $\Hard_k$ to retain coefficients at or above the $k$th largest magnitude,
including ties at the cutoff. Starting from $\mathbf{X}^{(0)} = \mathbf{Y}$, the iterative update is:
\begin{equation}
  \mathbf{X}^{(\kappa + 1)} \coloneqq P_{\Omega} \mathbf{Y} + (\mathrm{Id} - P_{\Omega} ) \,
  \mathrm{Re}\{\, \A_{ \bm{\theta} }^{-1} (\, \Hard_k ( \A_{ \bm{\theta} } ( \mathbf{X}^{(\kappa)} ))\,
  )\, \}\,. \label{eq:iht}
\end{equation}
This update alternates transform-domain hard thresholding with data consistency on the observed
pixels~\cite{blumensath2009, guleryuz2006}. The output matches observed entries but is not constrained to
remain $k$-sparse. The $K$th iterate yields $\widehat{ \mathbf{X} }_K( \bm{\theta}; \mathbf{Y}, \Omega)$
from \cref{eq:task}. Training requires derivatives with respect to $\bm{\theta}$; since these arise from
the iteration rather than a closed form, the iteration itself is differentiated. Fixing $K$ unfolds the
iteration into a finite computation graph---a piecewise differentiable map from gate angles to images,
which defines algorithm unrolling~\cite{monga2021}. Away from support changes, $\Hard_k( \mathbf{C} ) =
\mathbf{M}_{k} \odot \mathbf{C}$ for a fixed binary mask $\mathbf{M}_{k}$, with derivative
$D\Hard_k( \mathbf{C} )[\Delta \mathbf{C}] = \mathbf{M}_{k} \odot \Delta
\mathbf{C}$. Differentiation propagates through the retained support values via reverse mode over $K$
iterations, rematerializing at each step.

\subsection{Isometry without manifold optimization}
\label{sec:iso}

The fully relaxed network searches the product manifold $U(2)^{n} \times \big(U(1)^4\big)^{n(n-1)/2}$
per axis~\cite{pdft2026}. The diagonal relaxation pins the $U(2)$ factors at $\mathbf{H}$
and leaves a torus, parameterized by the periodic, redundant angles through $\theta\mapsto
e^{i\theta}$. Every Adam step~\cite{kingma2015} on these angles preserves gate unitarity, with no need
for Riemannian tangent-space projections and retractions (\cref{alg:train}). The diagonal gates at each
stage commute and combine into one diagonal, so applying the $n$ Hadamards and $n$ diagonals costs
$O(N^2\log N)$ per image, for the map or its adjoint, with no dense $\Net$ formed.

\subsection{Coherence is pinned}
\label{sec:pinned}

\begin{proposition}
  \label{prop:flat}
  Let $\mathbf{U} = \mathbf{B}\, \mathbf{G}_L \cdots \mathbf{G}_1$, where $\mathbf{B}$ is a permutation
  matrix and every factor $\mathbf{G}_l$ is either a unitary diagonal matrix or the Hadamard
  $\mathbf{H}_q$ of Step~1b, each wire carrying exactly one Hadamard factor. Then $\lvert
  (\mathbf{U})_{ij} \rvert = 1 / \sqrt{N}$, $\forall (i, j)$, so that $\sqrt{N}\, \mathbf{U}$ is a
  complex Hadamard matrix~\cite{tadej2006} with $\mu( \mathbf{U} ) = 1$. In particular $\mu = 1$ on QFT
  (diagonals) and its phase-only subfamily, tied (\cref{sec:shared}) or untied.
\end{proposition}

\begin{proof}
  The proof is omitted due to lack of space.
\end{proof}

\cref{prop:flat} guarantees minimum coherence throughout training via the circuit topology alone, without
requiring an explicit coherence penalty. Unconstrained learned dictionaries lack such guarantees. The QFT
(rotations) architecture falls outside this guarantee because its free parameters $\mathbf{V}_q$ replace
the prescribed Hadamards, allowing coherence to exceed the minimum bound of \num{1}---as observed in
\cref{sec:whichres}.

\section{Numerical Tests and Discussion}\label{sec:exp}

\begin{figure*}[t]
  \centering
  \includegraphics[width=.9\textwidth]{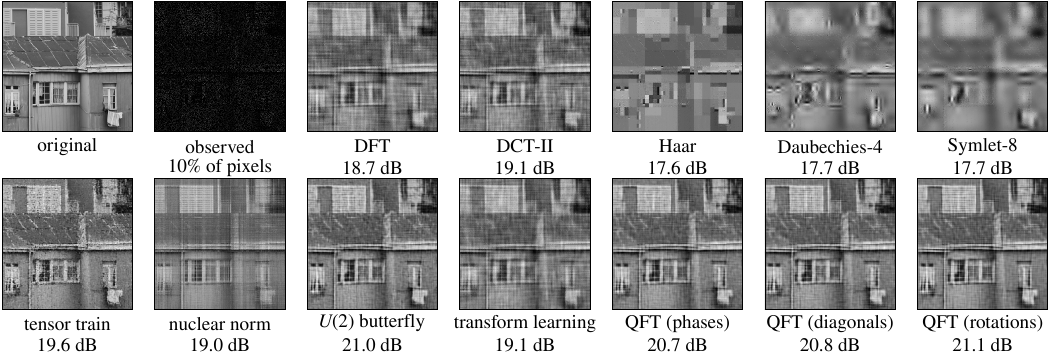}
  \caption{A DIV2K test image at sampling rate $p = 10\%$ and all methods in \cref{tab:reversal}, each at
    its PSNR-optimal budget.}
  \label{fig:example}
\end{figure*}

\subsection{Setup}

Numerical tests use DIV2K~\cite{agustsson2017}, grayscale $512 \times 512$ crops: \num{750} for training,
\num{50} for validation, and \num{100} test images. Each test image carries a single mask, drawn once
randomly and then fixed, shared by every method and every budget in this paper, and there is one training
seed throughout.

Each method is evaluated at its per-image budget optimum over sparsity levels $k/m \in \{0.015, 0.03,
0.0625, 0.125, 0.25, 0.5\}$, where $m$ is the observed pixel count. All transforms are evaluated at
$K=300$ iterations; the QFT models train at $K=100$, $T=200$ steps over mini-batches of two images,
$k=m/8$. Learning rates are tuned via validation-set PSNR over seven candidates, QFT (rotations)
also over five Cayley steps. In \cref{tab:reversal}, $\mu$ denotes coherence; ``Fitted on''
indicates what each method uses to train its parameters (nothing, test images, or training images). All
baselines are implemented in JAX~\cite{jax2018}.

\subsection{Comparison with fixed transforms and per-image methods}

\begin{table}[!b]
  \centering
  \caption{Completion on $100$ DIV2K test images at $512 \times 512$ and $p=10\%$: fixed bases, per-image
    fits, trained transforms, and proposed QFT networks}
  \label{tab:reversal}
  \setlength{\tabcolsep}{0.75pt}
  \newcommand{\hd}[2]{\multicolumn{1}{c}{\begin{tabular}[b]{@{}r@{}}#1\\#2\end{tabular}}}
  \resizebox{\columnwidth}{!}{%
    \begin{tabular}{@{}lrcrrrr@{}}
      \toprule
      Method & \hd{Trainable}{params} & \hd{Fitted}{on} & $\mu$ & PSNR & SSIM &
      \hd{MS-}{SSIM}\\
      \midrule
      DFT $=\Net( \bm{\theta}^0)$ & 0 & --- & {\bfseries 1} & 20.72 & 0.443 & 0.677\\
      DCT-II & 0 & --- & 4 & 21.15 & 0.454 & 0.693\\
      Haar \cite{mallat2008} & 0 & --- & 65536 & 18.70 & 0.453 & 0.582\\
      Daubechies-4 \cite{mallat2008} & 0 & --- & 68453 & 19.23 & 0.472 & 0.616\\
      Symlet-8 \cite{mallat2008} & 0 & --- & 95640 & 19.28 & 0.477 & 0.625\\
      \midrule
      Tensor train \cite{loeschcke2024} & [34--402]k & test img & --- & 22.13 & 0.524 & 0.791\\
      Nuclear norm \cite{candes2009matrix, jain2010svp, toh2010apg} & [1--109]k & test img & --- & 18.61 &
      0.369 & 0.612\\
      \midrule
      $U(2)$ butterfly \cite{dao2019} & 18432 & tr.\ data & 2.70 & 22.91 & 0.551 & 0.798\\
      Transform learning \cite{ravishankar2013closed} & 524288 & tr.\ data & 18642 & 21.03 & 0.453 & 0.691\\
      \midrule
      QFT (phases) & 72 & tr.\ data & {\bfseries 1} & 22.62 & 0.532 & 0.781\\
      QFT (diagonals) & 288 & tr.\ data & {\bfseries 1} & 22.74 & 0.537 & 0.786\\
      QFT (rotations) & 360 & tr.\ data & 1.20 & {\bfseries 22.98} & {\bfseries 0.558} & {\bfseries
        0.799}\\
      \bottomrule
    \end{tabular}
  }
\end{table}

QFT (diagonals) outperforms the DFT by $2.02$dB and the best fixed transform by $1.59$dB in PSNR,
and leads both in PSNR and MS-SSIM on all $100$ test images. \cref{fig:example} shows all methods on a
representative test image. The bases split by coherence into global (DFT-like) and localized
(wavelet-like) families, the global family leading on PSNR and MS-SSIM~\cite{wang2003}, with
single-scale SSIM~\cite{wang2004} favoring the wavelets.

Per-image methods are not transforms. The tensor train~\cite{loeschcke2024} trails by $0.60$dB in PSNR
but leads by $0.005$ in MS-SSIM. Among the trained transforms, the butterfly
factorization~\cite{dao2019} (blocks constrained to $U(2)$) leads by $0.17$dB at $64\times$ the
parameters, but its unconstrained variant loses $8$dB under a single-budget protocol; transform
learning~\cite{ravishankar2013closed} achieves $0.1$dB below the DCT.

\cref{fig:results}(a,\,b) compares methods across sampling rates, with the transforms trained at
$p=10\%$ applied unchanged. QFT (diagonals) leads the DFT at all rates on both metrics. The tensor
train leads in PSNR only at $1\%$ and in MS-SSIM up to $10\%$; the butterfly leads by $0.2$--$0.3$dB up
to $10\%$, matches at $20\%$, and trails above. \cref{fig:results}(c,\,d) analyzes computational cost.
QFT (diagonals) scales as $N^{1.7}$ and costs two-fifths of the tensor train at $512^2$ (which
scales as $N^{0.6}$; a rank cap). QFT (diagonals) requires $288$ trainable reals at $512^2$,
compared to $18{,}432$ (butterfly) and $524{,}288$ (transform learning). Training costs $741$s,
equivalent to $\sim 570$ tensor-train fits.

\begin{figure}[t]
  \centering
  \includegraphics[width=0.97\columnwidth]{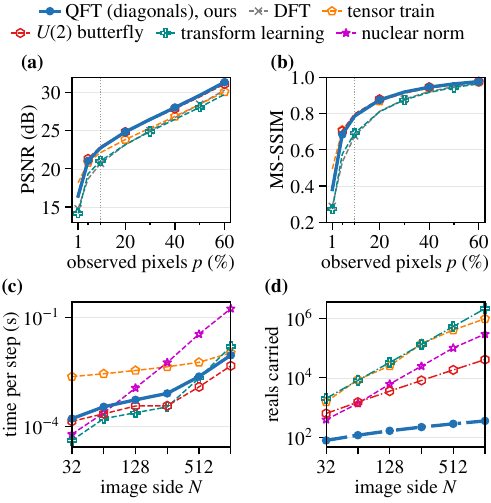}
  \caption{\textbf{(a)}, \textbf{(b)}~Mean PSNR and MS-SSIM over \num{100} test images vs.\ sampling
    rate (dotted: training rate). %
    \textbf{(c)}~Time per full-resolution step. %
    \textbf{(d)}~Reals each method carries, the transforms once, the per-image fits for every image.}
  \label{fig:results}
\end{figure}

\subsection{The coherence-sparsity trade-off}
\label{sec:hypothesis}

Theory provides a fundamental trade-off: a $k$-sparse object is recovered from $m$ random pixels when $m
\gtrsim \mu( \mathbf{U} )\, k\log N$~\cite{candes2007}. At fixed sampling, workable sparsity decreases
linearly with coherence. Wavelets achieve superior sparsity but suffer from high coherence---resulting in
worse inpainting PSNR than the less-sparse but low-coherence DFT. This reversal (wavelets win in
compression, lose in inpainting) cannot be fully attributed to coherence alone, since sparsity and
coherence vary together across bases. Isolating coherence's effect requires trainable transforms rather
than fixed bases. Standard coherence-reduction strategies also modify the sampling law~\cite{adcock2017},
which i.i.d.\ masks preclude.

The comparison uses fixed-$k$ hard-thresholding \cref{eq:iht}, applied uniformly to all methods. Soft
thresholding improves all bases---the DFT gains $1.0$dB at $10\%$ and $0.9$dB at $60\%$ sampling;
Symlet-8's high-rate deficit narrows from $3.0$ to $2.3$dB at $60\%$---but does not reverse the two
families' relative order.

\subsection{Freeing the QFT circuit: performance vs.\ structure}
\label{sec:whichres}

\cref{tab:reversal}'s final block of rows compares architectures with increasing degrees of freedom. The
phase-only model learns diagonal phase parameters while keeping Hadamards fixed. Freeing the diagonal
gates adds $+0.11$dB over phase-only updates. Freeing the Hadamards as well (replacing them with
learnable rotations, with both learning rates tuned) adds $+0.24$dB and $0.013$ MS-SSIM, but at a cost:
coherence rises with the Cayley step from $1.09$ to $1.31$, requiring Riemannian retractions at each
iteration. At the largest step attempted, training diverges ($\mu=697$, $12.5$dB). In contrast, the
diagonal subfamily (phases + fixed Hadamards) maintains $\mu = 1$ unconditionally with only $0.24$dB
loss, using plain Adam without retractions.

\section{Conclusion}
\label{sec:contentend}

The tension between sparsity and coherence in image inpainting motivated the design of trainable
transforms adapted to image structure. This work introduced quantum-inspired tensor-network circuits as
trainable transforms for the first time in this application. The diagonal QFT model achieved low
coherence through circuit topology alone, guaranteeing minimum coherence at every parameter value without
explicit penalties. This structural guarantee enabled efficient training with plain Adam and yielded
competitive performance with far fewer parameters than alternative learned transforms, demonstrating that
principled architectural choices can replace costly optimization constraints while establishing a new
paradigm of quantum-inspired methods to image reconstruction.

\clearpage
\section*{Acknowledgment}

S.~An's work was supported by JST SPRING, Japan, Grant Number JPMJSP2180. The authors thank \mbox{Jin-Guo
  Liu}, \mbox{Zhongyi Ni} and \mbox{Huanhai Zhou} of the Hong Kong University of Science and Technology
(Guangzhou) for discussions.

\balance
\printbibliography[heading=bibliography]

\end{document}